\documentclass[a4paper,11pt]{article}

\RequirePackage[T1]{fontenc}
\usepackage{amsfonts}
\usepackage{verbatim}
\usepackage{authblk}
\RequirePackage[numbers]{natbib}
\usepackage{natbib}
\RequirePackage[colorlinks, linkcolor=blue,citecolor=blue,urlcolor=blue, linktocpage]{hyperref}

\usepackage{amssymb, amsmath, natbib, amsthm, graphicx, euscript, mathrsfs,graphicx }

\newcommand{\RR}{\mathcal{R}}

\renewcommand{\u}{u^\circ}
\newcommand{\CC}{\mathcal{C}}
\newcommand{\F}{\varphi}
\newcommand{\G}{\psi}

\newcommand{\Var}
{\operatorname{Var}}

\newcommand{\E}{{\mathbb E}}

\renewcommand{\i}{\mathrm{i}}

\newcommand{\rrho}{\rho_{\u}}

\def\R{I\!\!R}

\newcommand{\N}{{\mathbb{N}}}
\renewcommand{\R}{{\mathbb{R}}}
\newcommand{\C}{{\mathbb{C}}}
\newcommand{\I}{{\mathbb I}}

\newcommand{\M}{{\mathcal{M}}}
\newtheorem{thm}{Theorem}[section]
\newtheorem{lem}[thm]{Lemma}
\newtheorem{prop}[thm]{Proposition}

 \theoremstyle{remark}
\newtheorem{rem}[thm]{Remark}
 \theoremstyle{remark}

\renewcommand{\Re}{\operatorname{Re}}

\usepackage{soul} % to highlight text
\usepackage{multicol}
\usepackage{multirow}
\usepackage{subcaption}
\usepackage{tabularx}

\begin{document}

\title{
Statistical Inference for Scale Mixture Models via Mellin Transform Approach\footnote{This version of the manuscript is accepted for publication in Statistics. The article was prepared in the framework of a research grant funded by the Ministry of Science and Higher Education of the Russian Federation (grant ID: 075-15-2022-325).}}

\author[1]{Denis Belomestny \thanks{denis.belomestny@uni-due.de}}
\author[2]{Ekaterina Morozova\thanks{eamorozova@hse.ru}}
\author[3]{Vladimir Panov\thanks{vpanov@hse.ru}}

\affil[1]{University of Duisburg-Essen, 
Thea-Leymann-Str. 9, 45127 Essen, Germany}
\affil[2,3]{HSE University\\
Laboratory of Stochastic Analysis and its Applications\\
             Pokrovsky boulevard 11, 109028 Moscow, Russia}
\affil[1]{IITP RAS}

%\thankstext{T1}{The article was prepared in the framework of a research grant funded by the Ministry of Science and Higher Education of the Russian Federation (grant ID: 075-15-2022-325).}

\maketitle

\begin{abstract}
This paper deals with statistical inference for the scale mixture models. We study an estimation approach based on the  Mellin -- Stieltjes  transform that  can be applied to both   discrete and absolute continuous mixing distributions. The accuracy of the corresponding estimate is analysed in terms of its expected pointwise error. As an important technical result, we prove the analogue of the Berry -- Esseen inequality for the Mellin transforms. The proposed statistical approach is illustrated   by numerical examples.
\end{abstract}

%\begin{keyword}[class=MSC]
%\kwd[Primary ]{60G51}
%\kwd{62M99}
%\kwd[; secondary ]{62F12}
%\end{keyword}

%\begin{keyword}
%\kwd{mixture models} 
%\kwd{scale mixtures} 
%\kwd{multiplicative censoring model}
%\kwd{Mellin -- Stieltjes transform}
%\kwd{Berry -- Esseen inequality}
%\kwd{Blumenthal-Getoor index}
%\end{keyword}

\section{Introduction}
In this paper we consider the problem of statistical inference for  multiplicative mixture models. More precisely, given a sample of i.i.d.\ random variables \(X_1,\dots, X_n\), \(n\in\N\), from the multiplicative mixture model of the form
\begin{equation}
\label{m1}
X=Y\eta,
\end{equation}
where \(Y\) and \(\eta\) are independent random variables, we aim at estimating the distribution of  one of these variables (say, \(Y\)) assuming that the law of another random variable (\(\eta\))  is known. For simplicity, we assume that both \(Y\) and \(\eta\) are almost surely positive, so that \(\eta\) can be viewed as the (stochastic) scaling parameter.

The aforementioned problem can be viewed as the problem of reconstructing the original signal from the contaminated sample and naturally arises in many applications. For instance, the case when \(\eta\) has a standard uniform distribution is known as the multiplicative censoring model and is widely employed in survival analysis~(Vardi, \citeyear{Vardi}). In this context, \(X\) corresponds to the time elapsed since the beginning of the disease, whereas \(Y\) represents the true survival time~(Van Es et al., \citeyear{VKO}). The precise estimation of the distribution of \(Y\) in this case would help the development of treatment programmes, as well as the assessment of their performance. Another example comes from finance, where the model~\eqref{m1} with  normally distributed  \(\eta\) and positive \(Y\) corresponds to the stochastic volatility model for describing  the log-returns of  an asset~(Van Es et al., \citeyear{VSV}, Belomestny and Schoenmakers,  \citeyear{BS15}).

While certain methods for estimation of the distribution of \(Y\) in model~\eqref{m1} already exist, they mostly assume some specific form of the distribution of \(\eta\). For instance, for the case when \(\eta\) follows a standard uniform distribution, some nonparametric estimation techniques are proposed by Vardi (\citeyear{Vardi}), Asgharian et al.\ (\citeyear{ACF}), Brunel et al. (\citeyear{BrCG}). Later, Comte and Dion (\citeyear{CD}) and Belomestny et al.\ (\citeyear{BCG}) generalise the setting and develop the estimators based on the projection techniques for the case when \(\eta\) follows the uniform distribution symmetric about one and the beta distribution, respectively. 

The problem of statistical inference for the multiplicative mixture models can be reduced to the additive deconvolution problem by taking logarithms of both parts in~\eqref{m1} or by taking logarithms of the squares in the alternating case. For the additive models, a wide range of estimation methods  is available; see, e.g., Zhang (\citeyear{ZhangFourier}), Meister (\citeyear{MeisterDeconv}), Belomestny and Goldenschluger (\citeyear{BG21}),  and numerous references therein. However, as was pointed out by Brunel et al.\ (\citeyear{BrCG}) and Belomestny and Goldenschluger (\citeyear{BG20}), this idea leads to several undesired consequences such as inability of estimation at zero and the loss of the information about the sign of the random variable.

For almost all papers mentioned above, the parametric assumption on the distribution of \(\eta\) is essential. The nonparametric case was considered by Belomestny and Goldenschluger (\citeyear{BG20}), Brenner Miguel, Comte and Johannes (\citeyear{MCJ}), Brenner Miguel and Phandoidaen (\citeyear{BMP}), where the kernel-type estimators based on the Mellin transform are introduced. However, these papers significantly employ the assumption of absolute continuity of the distributions of \(Y\) and \(\eta\).%, which makes their results inapplicable in case when either of the distributions is discrete or of mixed type. At the same time, the problem of recovering the distribution of \(X\) in a more general setting appears to be of interest.

In the current paper we do not restrict the class of distributions of \(\eta\) to be absolutely continuous or belong to a certain parametric family. The only assumption is that the set \(\mathcal{H}_G\) determined by~\eqref{HG} from the distribution of \(\eta\) is non-empty, and there exists a point \(\u\) such that the Mellin transforms of both distributions of \(Y\) and \(\eta\) are analytic. The latter assumption is rather typical for deconvolution problems, and the first one holds for a wide class of discrete distributions (see Section~\ref{secH} for the detailed discussion). Moreover, as we also show in the article, this assumption yields the parametric rate of convergence under very mild assumptions on the class of the probability density functions (p.d.f.s) of \(Y.\)

In order to avoid the assumption of absolute continuity of \(\eta,\) we formulate all results in terms of the distribution functions, assuming for simplicity that both random variables \(Y\) and \(\eta\) are a.s. positive. Note that the distribution function of \(X\) is equal to
\[
F_{\rm{mix}} (x) = \int_{\R^+} F\left(\frac{x}{\theta}\right)\,dG(\theta),
\]
where \(F\) and \(G\) are the cumulative distribution functions (c.d.f.s) of \(Y\) and \(\eta,\) respectively. 
Given  observations \(X_1,\ldots,X_n\) from \(F_{\rm{mix}}\), we aim to estimate the function \(F(x), x\in \R_+\), provided that  \(G(x), x\in \R_+,\) is known.

Our estimation method is based on the  Mellin -- Stieltjes transform, defined for a  function \(\F: \R_+ \to \R_+\), which is assumed to be a function of bounded variation over any bounded interval, as
\begin{equation}
\label{def_mellin}
\M[\F](z) := \int\limits_0 ^{\infty} x^{z-1}\,d\F(x), \quad z\in\C.
\end{equation}
The integral on the right-hand side is known to converge in a vertical strip \( \{z\in\C: \Re(z)\in [\alpha_\F, \beta_\F]=:\CC_\F\}\), with some \(\alpha_\F, \beta_\F>0\) (the degenerate case \(\alpha_\F=\beta_\F\) is also possible). 
Using the properties of the Mellin -- Stieltjes transform, we construct an estimator \(\widehat{F}\) of \(F\) (to be defined in Section~\ref{estimation_procedure}) and study the accuracy of this estimator at a fixed point \(x \in \R_+\) in  terms of 
the expected pointwise error 
\[
\RR^*(\widehat{F})=\RR^*(\widehat{F}; \u, x) :=
\E\Bigl[\bigl(\RR(\widehat{F}; \u, x)\bigr)^2\Bigr],
\]
where
\[
\RR(\widehat{F})=\RR(\widehat{F}; \u, x) := x^{\u-1}\bigl|F(x)-\widehat{F}(x)\bigr|
\] 
and \(\u\) is a technical parameter. We show that under rather mild assumptions on \(G\) and \(F\), the estimate \(\widehat{F}\) has \(1/n\) rate of convergence to \(F\) as measured in terms of  \(\RR^*(\widehat{F})\) with \(n\) being the sample size. 
\par
The paper is organised as follows. In the following section (Section~\ref{estimation_procedure}) we recall the most important properties of the Mellin -- Stieltjes transform and introduce the estimator \(\widehat{F}\) for \(F\). In Section~\ref{berry-esseen} we prove the analogue of the Berry -- Esseen inequality for the Mellin -- Stieltjes transforms (Lemma~\ref{lem1}), which plays an essential role for establishing the upper bounds for \(\RR(\widehat{F})\) and \(\RR^*(\widehat{F})\). The exact statements are given in Section~\ref{rates_of_convergence}, see Theorems~\ref{thm1} and~\ref{thm2}.
Next,  Section~\ref{convrates} contains the detailed discussion on the subclasses of distribution functions \(F\) and \(G\), for which the rate of convergence of the proposed estimators is polynomial. Section~\ref{secH} is devoted to one of the key assumptions of our estimation procedure, namely, that there exists a line parallel to the imaginary axis such that \(\M[G](z)\neq 0\) for any \(z\) on this line. It is shown that the aforementioned assumption is fulfilled, in particular, for any discrete positive distribution separated from zero. Finally, Section~\ref{numerical_examples} contains a numerical example demonstrating the performance of our estimator via a simulation study. Appendix~\ref{exBE} contains a numerical example of the application of the Berry -- Esseen inequality. All proofs are collected in Appendix~\ref{AppB}.

% \cite{BP18}

\section{Estimation procedure}
\label{estimation_procedure}
Note that the Mellin -- Stieltjes transform of \(F_{\rm{mix}}\) is equal to \[\M[F_{\rm{mix}}](z) := \E[X^{z-1}] = \M[F](z)\cdot \M[G](z) \] for any \(z\in\C\) such that both Mellin transforms on the right-hand side are well defined. 
Since we assume that the distribution of \(\eta\) is known, we can estimate the Mellin transform of \(F\) by
\[\widehat{\mathcal{M}[F]}(z) := \frac{1}{n}\sum\limits_{i=1} ^n X_i ^{z-1}/\M[G](z),\]
at any point \(z \in \C\) such that \(\M[G](z) \ne 0.\) Let us introduce the notation
 \begin{eqnarray}\label{HG}
\mathcal{H}_G:=\bigl\{
 u \in \CC_G: \; \M[G](u+\i v) \ne 0 \quad \forall \; v \in \R
 \bigr\}.\end{eqnarray}
 The set \(\mathcal{H}_G\) is non-empty in most cases; for instance, as we show in  Section~\ref{secH}, for positive discrete  distributions, there exists some \(\tilde{u}<1\) such that \((-\infty, \tilde{u}) \subset \mathcal{H}_G,\) provided that the distribution is separated from 0. %Motivated by this example, we will concentrate on the case when the set \(\mathcal{H}_G\) includes at least one number smaller than \(1.\) 
 
The estimator of \(F\) is based on 
the inversion formula for the Mellin transform. Some versions of this formula are known in the literature, see, e.g., Section~7 from~\cite{Kawata}, but here we need  a slightly different form.

\begin{lem}\label{lemm}
Let \(\F:\R_+ \to \R\) be a non-decreasing function. %satisfying \(F(0)=0\) and \(\lim\limits_{x\to\infty} F(x)/x=0\). 
Then 
\begin{enumerate}
\item  if there exists some \(\u \in (-\infty,1) \cap \CC_\F\), then \(\F(0)<\infty\) and
\[ \frac{1}{2\pi i} \int\limits_{\u-i\infty} ^{\u+i\infty} x^{-z+1} \frac{\M[\F](z)}{-(z-1)}\,dz = \frac{1}{2}\Bigl( \F(x+0)+\F(x-0) \Bigr) - \F(0);\]
\item  if there exists some \(\u \in (1, \infty) \cap \CC_\F\), then \(\F(\infty)<\infty\) and
\[ \frac{1}{2\pi i} \int\limits_{\u-i\infty} ^{\u+i\infty} x^{-z+1} \frac{\M[\F](z)}{-(z-1)}\,dz = \frac{1}{2}\Bigl( \F(x+0)+\F(x-0) \Bigr) -\F(\infty).\]
\end{enumerate}
\end{lem}
\begin{proof}
The proof is given in Appendix~\ref{A1}.\end{proof}
\begin{rem}
In what follows, for any non-decreasing function \(\F:\R_+ \to \R\), we will use the same notation \(\F(x), x\in\R_+,\) for the standardised version of the function, that is, for \(\bigl(\F(x+0)+\F(x-0)\bigr)/2.\)
\end{rem}
Motivated by Lemma~\ref{lemm}, 
we define the estimator \(\widehat{F}(x)\) of \(F(x)\) as
\begin{equation}
\label{est_main}
\widehat{F}(x) := \frac{1}{2\pi} \int\limits_{-\infty} ^{\infty} x^{-\u-\i v+1} \frac{\widehat{\M[F]}(\u+\i v)}{-(\u+\i v-1)}K(v)\,dv
\end{equation}
for some \(\u \in (-\infty,1) \cap \CC_F \cap \mathcal{H}_G\), and the kernel function \(K\) of the form 
\begin{eqnarray}
\label{est_main2} K(x) := \left(1-\frac{|x|}{T}\right)\mathbb{I}\{|x|\leq T\},
\end{eqnarray}
with  a  positive number \(T\). This choice of the kernel function is inspired by the Berry -- Esseen inequality for the Mellin transforms. In the next section, we discuss this inequality.

\section{Berry -- Esseen inequality for the Mellin transforms}
\label{berry-esseen}

The following lemma is motivated by a similar result for the Fourier transform, see Section~4.1 from  \cite{LinBai}. 
\begin{lem}
\label{lem1}
Let \(\F, \G: \R_+ \to \R_{+}\) be two non-decreasing left continuous  functions such that \(\F(0)=\G(0)=0\). Let \(\u\in (-\infty,1)\cap \CC_\F \cap \CC_\G\). Denote 
 \begin{eqnarray*}
\rrho(\F,\G):= \sup\limits_{x \geq 0} \bigl|x^{\u-1} \bigl( \F(x)-\G(x) \bigr)\bigr|.
\end{eqnarray*}
Assume that the supremum in the definition of \(\rrho(\F,\G)\) is attained at some point \(x_0>0\).
Then for any \(b>2/\pi\), it holds
\begin{multline}
\label{berry}
\rrho (\F,\G) \leq \frac{b}{2}\int\limits_{-T} ^T \frac{|\M[\F](\u+\i v)-\M[\G](\u+\i v)|}{|v|}\,dv \\+ bT x_0^{\u-1} \int_0^{2c(b)/T}\left|\G(x_0)-\G\left(x_0e^r\right)\right|\,dr,
\end{multline}
where \(c(b)\) can be found as a unique root of the equation
\begin{equation}
\label{cb}
\int\limits_{|r|\leq c(b)} \frac{\sin^2 r}{\pi r^2}\,dr 
=\frac{2}{3}\left(1+\frac{1}{\pi b}\right),\end{equation}
and 
\(T\) is an arbitrary positive number such that
\begin{eqnarray}\label{T}
T>2 c(b)(1-\u)/\log 2.
\end{eqnarray}
% and
%\begin{equation}
%\label{cb2}
%\int\limits_{|r|<c(b)} \frac{\sin^2 r}{r^2} e^{r(1-\u)/T}\,dr =  \frac{1}{2b} + \frac{\pi}{2},
%\end{equation}
%provided \(\u\geq 1.\)

\end{lem}
\begin{proof}
The proof is given in Appendix~\ref{A}.
\end{proof}
\begin{rem}
Since the left-hand side of the equation \eqref{cb} tends to one as \(c(b) \to \infty\), and to zero as \(c(b) \to 0\), it can take any values from the interval \((0,1).\) This is the reason for the restriction \(b>2/\pi\).
\end{rem}

\begin{rem}
\label{rem1}
It is a worth mentioning that the second term in the upper bound~\eqref{berry} tends to zero as \(T\to\infty\) at a polynomial rate in the case when \(\G\) is uniformly \(\alpha\)-H\"older continuous with \(\alpha\in(0,1]\).\footnote{The asymptotic behaviour of the first term in~\eqref{berry} will be discussed later. Definitely, it depends on the closeness between the functions \(\F\) and \(\G.\)} Indeed, we have
\begin{equation*}
\int\limits_0 ^{2c(b)/T} |\G(x_0)-\G(x_0e^r)|\,dr \leq L_{\psi} x_0^{\alpha}\int\limits_0 ^{2c(b)/T} |1-e^r|^{\alpha}\,dr
\end{equation*}
for some \(L_{\psi}\in(0,\infty)\). Using the  inequality \(e^x-1\leq x e^x\) which holds for any \(x>0,\) 
we get 
\begin{eqnarray*}
\int\limits_0 ^{2c(b)/T} |1-e^r|^{\alpha}\,dr \leq e^{2\alpha c(b)/T}\frac{(2c(b))^{1+\alpha}}{(1+\alpha) T^{1+\alpha}}.
\end{eqnarray*}
Hence, 
\[
bT x_0 ^{\u-1} \int_0^{2c(b)/T}\left|\G(x_0)-\G\left(x_0e^r\right)\right|\,dr \leq L_{\psi} x_0 ^{\alpha+\u-1}e^{2\alpha c(b)/T}\frac{b(2c(b))^{1+\alpha}}{(1+\alpha) T^{\alpha}},
\] 
yielding the polynomial rate of decay.
\end{rem}

\begin{rem} For the case considered in Remark~\ref{rem1}, one can show a slightly different form of the Berry -- Essen inequality~\eqref{berry}, namely,
\begin{multline*}
\rrho (\F,\G) \leq \frac{b}{2}\int\limits_{-T} ^T \frac{|\M[\F](\u+\i v)-\M[\G](\u+\i v)|}{|v|}\,dv \\+ bT \int_0^{2c(b)/T}\left|x_0^{\u-1} \G(x_0)- (x_0 e^r)^{\u-1} \G\left(x_0e^r\right)\right|\,dr \bigl(1 +O(1)).
\end{multline*}
We have
\begin{multline*}
T  \int_0^{2c(b)/T}\left|\G(x_0)-
\G\left(x_0e^r\right)\right|\,dr\\
\leq 
T  \int_0^{2c(b)/T}\bigl|\G(x_0)-e^{r(u^\circ-1)}\G\left(x_0e^r\right)\bigr|\,dr
+
T  \int_0^{2c(b)/T}  \bigl( 1 - 
e^{-r(1- u^\circ)}  \bigr)\, \G\left(x_0e^r\right)dr.  
\end{multline*}
The order of the second term is \(O(T^{-1})\), since 
\begin{multline*}
\Bigl| T  \int_0^{2c(b)/T}  \bigl( 1 - 
e^{-r(1- u^\circ)}  \bigr)\, \G\left(x_0e^r\right)dr 
\Bigr| 
\\ 
\leq
T  x_0^\alpha \int_0^{2c(b)/T}  \bigl( 1 - 
e^{-r(1- u^\circ)}  \bigr)\, e^{r \alpha} dr \hspace{1.2cm}\\
 \leq
T  x_0^\alpha \Bigl[ \frac{e^{ 2 \alpha c(b) /T}-1}
{\alpha}
-\frac{e^{2 (\alpha-(1-u^\circ)) c(b) /T}-1}{\alpha-(1-u^\circ)}
\Bigr] 
= O(T^{-1}),
\end{multline*}
while the first summand is of order \(O(T^{-\alpha}) \),
since 
\begin{multline*}
T  \int_0^{2c(b)/T}\bigl|\G(x_0)-e^{r(u^\circ-1)}\G\left(x_0e^r\right)\bigr|\,dr
\\\leq 
T  \int_0^{2c(b)/T}\left|\G(x_0)-
\G\left(x_0e^r\right)\right|\,dr
+
T  \int_0^{2c(b)/T}  \bigl( 1 - 
e^{-r(1- u^\circ)}  \bigr)\, \G\left(x_0e^r\right)dr\\
= O(T^{-\alpha}),  
\end{multline*}
see Remark~3.3.
\end{rem}
The numerical example of  the use of the Berry -- Esseen inequality is given in Appendix~\ref{exBE}.
\section{Main results}
\label{rates_of_convergence}
In this section, we provide the rates of convergence  of the estimate~\eqref{est_main} with the kernel~\eqref{est_main2}. 
The following theorem holds.
\begin{thm}
\label{thm1}
 Let 
\(\u\in (-\infty,1) \cap \CC_F \cap \mathcal{H}_G.\) Assume that the distribution function \(F\) is continuous at least in a small vicinity of \(0\), and 
\begin{eqnarray}\label{add}
x^{\u-1} F(x) \to 0 \qquad \mbox{as} \quad x \to 0.
\end{eqnarray}
Let \(b>2/\pi\) be an arbitrary number,  \(c(b)\) be the solution of the equation~\eqref{cb} and \(T\) satisfy~\eqref{T}. Then there exists some \(x_0>0\) not depending on \(T\) and \(n\) such that for any fixed \(x>0,\) it holds with probability \(1\)
\begin{eqnarray}
\nonumber
\RR(\widehat{F}; \u, x) &\leq& \frac{b}{2T} %\frac{1}{T}\int\limits_{-T} ^T e^{-iv\log x} \frac{\M[F](1+iv)|v|}{-(1+iv)}\,dv + \int\limits_{|v|>T} e^{-iv\log x} \frac{\M[F](1+iv)}{-(1+iv)}\,dv\right| 
 \int\limits_{-T} ^{T}  \left| \M[F](\u+\i v)\right| \,dv \\ && \nonumber \hspace{1cm}+ bT x_0^{\u-1}  \int\limits_0 ^{2c(b)/T} \left|F(x_0)-F\left(x_0e^r\right)\right|\,dr\\ && \nonumber \hspace{1cm}+ \frac{x^{\u-1}}{2\pi n}\left|\sum\limits_{k=1} ^n \Lambda (X_k, x)\right|, \label{f1}
\end{eqnarray}
where
\begin{eqnarray*}\Lambda (X_k, x) = 
\int\limits_{-\infty} ^{\infty}x^{-(\u-1)-\i v} \frac{\M[F_{\rm{mix}}](\u+\i v)-X_k ^{(\u-1)+\i v}}{-(\u+\i v-1)\M[G](\u+\i v)}K(v)\,dv.%\int\limits_{-\infty} ^{\infty}x^{iv}\frac{\M[F_{mix}](1+iv)-X_k ^{iv}}{-(1+iv)\M[G](1+iv)}\left(1-\frac{|v|}{T}\right)\mathbb{I}\{|v|\leq T\}\,dv,
\end{eqnarray*}
%or \eqref{cb2}, depending on \(\u<1\) and \(\u\textcolor{red}{\geq}1.\)
\end{thm}
\begin{proof}
The proof is given in Appendix~\ref{B}.
\end{proof}
The preceding theorem allows us to further analyze the convergence of the estimator \(\widehat{F}\) to the true c.d.f. \(F\) in the \(\mathcal{L}^2\)-sense. 
\begin{thm}\label{thm2}
Under the same notations and assumptions as in Theorem~\ref{thm1}, we get for any fixed \(x>0,\)
%\ty{Let \(\u\in (-\infty,1) \cap \CC_F \cap \mathcal{H}_G.\)} Then for any \ty{fixed} \(x>0\),
\begin{eqnarray}
\nonumber
\RR^*(\widehat{F}; \u, x) &\leq& \frac{3b^2}{4T^2} \left(\int\limits_{-T} ^T |\M[F](\u+\i v)|\,dv\right)^2 \\ \nonumber &&\hspace{1cm}+3b^2 T^2 x_0^{2(\u-1)} \left(\int\limits_0 ^{2c(b)/T} \left|F(x_0)-F\left(x_0 e^r\right)\right|\,dr\right)^2 \\ \nonumber && \hspace{1cm}+\frac{3}{4\pi^2 n} \int\limits_{-T} ^T \frac{1}{((\u-1)^2 + v^2)|\M[G](\u+\i v)|^2} \,dv \hspace{0.3cm}\\ && \hspace{2.5cm} \times\int\limits_{\R} |\M[F_{\rm{mix}}](2\u-1+\i w)|\,dw.\label{l2bound}
\end{eqnarray}
%where, as before, \(b>2/\pi\), and \(c(b)\) can be found as the root of the equation~\eqref{cb}.
%and~\eqref{cb2} depending on \(\u<1\) and \(\u\geq 1\) as before.
\end{thm}
\begin{proof}
The proof is given in Appendix~\ref{C}.
\end{proof}

\section{Convergence rates} 
\label{convrates} 
Denote the summands in~\eqref{l2bound} by \(I_1, I_2, I_3,\) respectively. Let us consider these summands separately. 
\par
\textbf{1.} Let us first note that the condition 
 \begin{eqnarray}\label{as}
 \int\limits_{\R} |\M[F](\u+\i v)|\,dv <\infty\end{eqnarray}
is fulfilled for many distributions, including the distributions with exponential and polynomial decay of the Mellin transforms, see~\cite{BP15}. The condition~\eqref{as} yields that \textbf{the first summand} is of order \(I_1 = O(T^{-2})\) as \(T \to \infty.\) Nevertheless, it can be shown that any distribution satisfying~\eqref{as} is absolutely continuous.  This fact can be proved in various ways, e.g., via the convergence in \(\mathcal{L}^2(\R)\), along the same lines as the proof of the similar fact for the Fourier transform, see Lemma~1.1 from \cite{FP}. It can be also derived from the properties of the Fourier transform (see Theorem~3.2.2 from~\cite{Lukacs}). 
In fact, \begin{eqnarray*}
\int\limits_{\R} \Bigl|\M[F](\u+\i v)\Bigr|\,dv
&=& 
\int\limits_{\R} \Bigl|\int_{\R+} x^{\u+\i v-1 }   dF(x)\Bigr|\,dv\\
&=&
\int\limits_{\R} \Bigl|\int_{\R} e^{\i y v} \cdot e^{y(\u-1)} dF(e^y)\Bigr|\,dv\\
&=& 
\M[F](\u)
\int\limits_{\R} \Bigl|\int_{\R} e^{\i y v} dH(y)\Bigr|\,dv,
\end{eqnarray*}
where 
\begin{eqnarray*}
H(y) &:=& \frac{1}{
\M[F](\u)
}\int_{-\infty}^y e^{s(\u-1)} dF(e^s) \\ 
&=&
\frac{1}{
\M[F](\u)
}
 \int_0^{e^y} x^{\u-1} dF(x), \qquad\qquad y \in \R,
 \end{eqnarray*}
is a distribution function (here we use that \(\M[F](\u)<\infty\) since
\(\u \in \mathcal{C}_F\)). Therefore, the condition~\eqref{as} is equivalent to \(\mathcal{F}[H] \in \mathcal{L}^1(\mathbb{R})\), which yields that the \(H\) is a c.d.f. of an absolutely continuous distribution with bounded and continuous density function. Then the distribution of \(F\) is absolutely continuous as well.
 \par
\textbf{2.} The order of \textbf{the second summand} is also \(O(T^{-2})\) in some cases. For instance, whenever \(F\) is Lipschitz continuous, we get from Remark~\ref{rem1}  
\begin{eqnarray*}
I_2\leq L_F^{2} x_0^{2\u}e^{4c(b)/T}\frac{12b^2(c(b))^4}{T^2}
\end{eqnarray*}
where \(L_F\in(0,\infty)\) is the Lipschitz constant of \(F\).\newline 
\par
\textbf{3.} To establish the asymptotic order of \textbf{the third summand} \(I_3\), we observe that, since \(\u\in \mathcal{H}_G \subset \CC_G\), for any \(T>0,\)
there exist positive constants \(C_{\u,T} ^{(1)},C_{\u,T} ^{(2)}\) such that 
\begin{equation}
\label{A3}
C_{\u, T}^{(1)} \leq |\M[G](\u+\i v)| \leq C_{\u, T}^{(2)}, \quad \forall v\in [-T,T].
\end{equation}
 Assuming that \(2 \u -1 \in \mathcal{H}_G,\) we similarly get with some positive constants \(C_{2\u-1,T} ^{(1)},C_{2\u-1,T} ^{(2)}\)
\begin{equation}
\label{A31}
C_{2\u-1, T}^{(1)} \leq |\M[G](2\u-1+\i v)| \leq C_{2\u-1, T}^{(2)}, \quad \forall v\in [-T,T].
\end{equation}
Let us assume additionally that as \(T\to\infty\), 
\begin{eqnarray}
\label{assa}
C_{\u, T}^{(1)} \to C^{(1)}_{\u} >0 \quad \mbox{ and  } \quad  C_{2\u-1, T}^{(2)} \to C_{2\u-1}^{(2)} <\infty.
\end{eqnarray}
In particular, it can be noted that the upper bound is always satisfied with \(C_{2\u-1} ^{(2)} \leq \E[\eta^{2(\u-1)}]<\infty\). As for the lower bound, we refer to Section~\ref{secH} for  examples of  mixing distributions which satisfy this assumption.
Finally, assume that the analogue of~\eqref{as} holds along the line \(\Re(z)=2\u-1\),
 \begin{eqnarray}\label{as2}
 \int\limits_{\R} |\M[F](2\u-1+\i v)|\,dv \leq C_1<\infty.\end{eqnarray}
Then we get
\begin{equation*}
I_3 \leq \frac{3x^{2(\u-1)}}{4\pi^2 n}  \frac{C_1 C_{2\u-1} ^{(2)}}{\left(C_{\u} ^{(1)}\right)^2}  \int\limits_{-T} ^T \frac{1}{(\u-1)^2+v^2}\,dv = \frac{3 x^{2(\u-1)} C_1 C_{2\u-1} ^{(2)}}{4(1-\u)\pi \left(C_{\u} ^{(1)}\right)^2 n}.
\end{equation*}
Finally,  let us note that the choice \(T\gtrsim \sqrt{n}\), leads to the rate of convergence \(I_1+I_2+I_3 = O(1/n)\) which  coincides with the rate  in the case when the direct observations from \(F\) are available.
The above discussion  leads to the following result.
\begin{prop}
Assume that there exists at least one \(\u\in (-\infty,1) \cap \CC_F \cap \mathcal{H}_G\)  such that the assumptions~\eqref{add}, \eqref{A31} and \eqref{assa} are fulfilled. If the distribution of \(Y\) is absolutely continuous, then under any choice \(T\gtrsim \sqrt{n}\), it holds
\begin{eqnarray*}
\RR^*(\widehat{F}; \u, x)  \lesssim 1/n
\end{eqnarray*}
as \(n \to \infty.\)
\end{prop}

\section{The set \(\mathcal{H}_G\) 
}\label{secH}
The set \(\mathcal{H}_G\) defined by~\eqref{HG} plays a crucial role in the analysis. Let us show that this set is nonempty for positive discrete distributions separated from \(0\). 
Indeed, consider a random variable \(\eta\) taking \(K\) positive values \(\sigma_1,\dots,\sigma_K\) with probabilities \(p_1,\dots,p_K >0\), \(\sum_{k=1} ^K p_k=1\) (the analysis of the case \(K=\infty\) follows the same lines). Without loss of generality, assume that \(0<\sigma_1 < \sigma_2 <\dots<\sigma_K\). Then  it holds
%\begin{eqnarray*}
%\bigl|
%\sum_{k=1}^K \sigma_k^{z-1} p_k
%\bigr| 
%\leq 
% \sum_{k=1}^K \sigma_k^{\Re(z)-1} p_k  \leq \sigma_1^{\Re(z)-1},
%\end{eqnarray*}
%giving the upper bound in~\eqref{A2}. As for the lower bound, one can observe that 
\begin{eqnarray}\label{e1}
\M[G](u+\i v)=
\sum_{k=1}^K \sigma_k^{u+\i v-1} p_k  
=
\sigma_1^{u+\i v-1} p_1 + 
\sum_{k=2}^K \sigma_k^{u+\i v-1} p_k,
\end{eqnarray}
where  the modulus of the second summand can be upper bounded as   
\begin{eqnarray*}
\left|
\sum_{k=2}^K \sigma_k^{u+\i v-1} p_k
\right| 
\leq \sum_{k=2}^K | \sigma_k^{u+\i v-1}| p_k
= \sum_{k=2}^K \sigma_k^{u-1} p_k
< \sigma_2^{u-1} (1-p_1)
\end{eqnarray*}
for any \(u<1,\) while the modulus of the first summand in \eqref{e1} is equal to \(  \sigma_1^{u-1} p_1.\) From this it follows that 
\begin{eqnarray*}
\bigl|
M[G](u+\i v)\bigr|
%\bigl|
%\sum_{k=1}^\infty \sigma_k^{s-1} p_k  
%\bigr|
\geq 
\bigl|
\sigma_1^{u+\i v-1} p_1
\bigr| 
-
\Bigl|
\sum_{k=2}^K \sigma_k^{u+\i v-1} p_k
\Bigr|
> 
 \sigma_1^{u-1} p_1 - \sigma_2^{u-1} (1-p_1),
\end{eqnarray*}
where the expression on the right-hand side is strongly positive for any 
%\begin{eqnarray*}
%\sigma_2^{u-1} < \sigma_1^{u-1} p_1,
%\end{eqnarray*}
%or, equivalently,
\begin{eqnarray}\label{ff}
u< \frac{\log(p_1/ (1-  p_1))}{\log(\sigma_2/\sigma_1)}+1.
\end{eqnarray}
Let us provide a couple of examples.
\begin{enumerate}
\item Consider the geometric distribution defined as a distribution of a r.v. \(\eta\) such that 
\begin{eqnarray*}
\mathbb{P}\bigl\{
	\eta=k
\bigr\}
=
p (1-p)^{k-1}, \qquad k=1,2,...
\end{eqnarray*}
According to~\eqref{ff}, the absolute value of the Mellin transform of this distribution is strongly positive for any 
\begin{eqnarray*}
u< \frac{\log\Bigl(p/(1-p)\Bigr)}{\log(2)} +1,
\end{eqnarray*}
where the expression on the r.h.s. is positive for \(p>1/3\) and negative for \(p<1/3.\)
\item Consider an analogue of the Poisson distribution on \(\{1,2,...\}:\)
\begin{eqnarray*}
\mathbb{P}\bigl\{
	\eta=k
\bigr\}
=
\frac{e^{-\lambda}}{1-e^{-\lambda}} \cdot \frac{\lambda^k}{k!}, \qquad k=1,2,3,...,
\end{eqnarray*}
where \(\lambda>0.\) Returning to~\eqref{ff}, we get that the absolute value of the Mellin transform  is strongly positive in absolute value for any
\begin{eqnarray*}
u< \frac{\log\Bigl( \lambda e^{-\lambda} / (1- e^{-\lambda} -\lambda e^{-\lambda})\Bigr)}{\log(2)} +1,
\end{eqnarray*}
which is  positive whenever \(\lambda \in (0, \lambda_\circ)\) with \(\lambda_\circ\) defined as the positive solution of the equation \(3\lambda+1= e^\lambda.\) Numerically we get
 \(\lambda_\circ \approx 1.9 \).

\end{enumerate}
Note that in some simple cases the result can be enhanced. 
\begin{enumerate}
\item For instance, consider the zeta distribution 
\begin{equation}
\label{zetadistr}
\mathbb{P}\{\eta=k\}=\frac{1}{\zeta(s)}k^{-s}, \quad k=1,2,\dots,
\end{equation}
where \(s>1\) and \(\zeta(s)\) is the Riemann zeta function defined as \[\zeta(s)=\sum\limits_{k=1} ^{\infty} k^{-s}.\] Since \(\M[G](u+\i v)=\zeta(1+s-u-\i v)/\zeta(s)\), and the zeta-function has no zeros with real part larger than 1, we get that the set \(\mathcal{H}_G\) contains all 
\(u<s.\)
\item 
Consider the distribution of  the random variable taking values \(\sigma_1=1\) and \(\sigma_2>1\) with probabilities \(p \in (0,1)\) and \(1-p\). Then any 
\begin{equation}
\label{u_circ}
u \ne  u_{\circ}:=1+\log_{\sigma_2}(p/(1-p))
\end{equation} 
belong to the set \(\mathcal{H}_G\). In fact, in this case the Mellin transform can be lower bounded for any \(v\in\R\) by
\begin{eqnarray*}
\left| \M[G](u+\i v)\right| \geq \bigl|
p- (1-p) \sigma_2^{u-1}
\bigr|,
\end{eqnarray*}
which is strongly larger than \(0\) for any \(u \ne u_\circ.\)
\end{enumerate}
%\ty{
%\section{Application to the alternating models}
%In this section, we show that the proposed method can be adapted for more complicated models. For instance, let us consider the alternating model 
%\begin{eqnarray*}
%X=Y_1 \zeta_1 \eta - Y_2 \zeta_2 (1-\eta),
%\end{eqnarray*}
%where \(Y_1, Y_2, \zeta_1, \zeta_2\) are a.s. positive random variables, \(Y_1 \eqd Y_2, \zeta_1 \eqd \zeta_2,\) \(\eta\) is a Bernoulli random variable with \(\PP\{\eta=1\}=p,\) and the variables \(Y_1, Y_2, \zeta_1, \zeta_2, \zeta\) are jointly independent. Then 
%\begin{eqnarray*}
%|X| \eqd Y_1 \zeta_1,
%\end{eqnarray*}
%and the same methodology can be applied.
%}
\section{Numerical examples}
\label{numerical_examples}

In this section we use simulated data to illustrate the behaviour of the proposed estimator~\eqref{est_main} with the kernel~\eqref{est_main2} for different types of  mixing distributions \(G\). We also compare our estimator with another one based on  a logarithmic transformation  of both sides in~\eqref{m1} and interpretation of the problem as an additive deconvolution problem. Since \(\phi_{\log X}(t) = \phi_{\log Y}(t) \phi_{\log \eta}(t)\), \(\forall t\in\R\) where \(\phi_{X}\) denotes the characteristic function of \(X\), the inversion formula for the Fourier-Stieltjes transform (see Theorem 4.4.1 in Kawata, \citeyear{Kawata}) suggests that \(F\) can be estimated using the relation \(F (x) = F_{\log Y} (\log x)\) by
\begin{equation}
\label{fourier}
\widehat{F}_{\log Y}(\log x) - \widehat{F}_{\log Y}(0) = \frac{1}{2\pi}\int\limits_{-R} ^{R} \frac{e^{-\i t \log x}-1}{-it} \cdot \frac{\widehat{\phi}_{\log X} (t)}{\phi_{\log \eta} (t)} \phi_{\kappa}(ht)\,dt,
\end{equation}
where \(\widehat{\phi}_{\log X} (t)=\frac{1}{n}\sum_{j=1} ^n e^{\i t \log X_j}\), \(R\to\infty\) as \(n\to\infty\), \(h\) is a bandwidth parameter chosen by the cross-validation approach, and \(\phi_{\kappa}\) is the characteristic function of some kernel function \(\kappa\) which we choose to be  density of the standard normal law; for details see, e.g., Delaigle (\citeyear{Delaigle}). The value of \(\widehat{F}_{\log Y}(0)\) can be recovered as the minimal value which, when added to the right-hand side of~\eqref{fourier}, would yield a non-negative value for all \(x\in\R\), since \(F(x)\) only takes values in \(\R_+\). In the subsequent examples we simply replace \(\widehat{F}_{\log Y} (0)\) with the true value \(F_{\log Y} (0) = F(1)\).
\par
In what follows, we assume  that the unknown distribution \(F\) is absolutely continuous. More precisely, 
consider the cases when \(F\) is a c.d.f.\ of the beta distribution
\[F(x) = \frac{\I\{x\in[0,1]\}}{B(\alpha_1,\alpha_2)}\int\limits_{0} ^x y^{\alpha_1-1} (1-y)^{\alpha_2-1}\,dy  + \I\{x>1\}\]
 and the gamma distribution
 \[F(x) = \frac{\theta^k}{\Gamma(k)} \int\limits_0 ^{x}  y^{k-1} e^{-\theta y}\,dy \cdot  \I\{x>0\}\]
with  parameters \(\alpha_1=\alpha_2=k=\theta=2\). 
Note that \(\mathcal{C}_F = (1-\alpha_1,\infty)\) for the beta distribution and \(\mathcal{C}_F = (1-k,\infty)\) for the gamma distribution. Also note that the assumption~\eqref{add} holds for both examples: in fact, since \(F(0)=0,\) we get by the L'H\^opital rule
\begin{eqnarray*}
\lim_{x \to 0} x^{\u-1}F(x) 
=
\frac{1}{1-\u}
\lim_{x \to 0} x^{\u}F'(x),
\end{eqnarray*}
where the limit on the right-hand side is finite for any \(\u \in (1-\alpha_1,1)=(-1,1)\) in the case of the beta distribution, and for any \(\u \in (1-k,1)=(-1,1)\) in the case of the gamma distribution.

As for the mixing distribution \(G\), we consider the following three  cases.
\begin{enumerate}
\item \underline{Discrete mixing distribution with finite support.} Let \(G\) be a purely discrete probability law taking values 1 and 2 with probabilities 1/3 and 2/3, respectively. As was discussed in Section~\ref{secH}, we have in this case
\begin{eqnarray*}
|\M[G](u+\i v)|\ne 0 \quad \mbox{ for any } \hspace{0.235cm}
u\neq 1+\log_2 \frac{1/3}{2/3}=0 \hspace{0.235cm} \mbox{ and any} \hspace{0.235cm} v \in \R.
\end{eqnarray*}
Therefore, the parameter \(\u\) can be taken from the set 
\begin{eqnarray*}
(-\infty,1) \cap \CC_F \cap \mathcal{H}_G = (-1,0) \cup (0,1).
\end{eqnarray*}
\item \underline{Discrete mixing distribution with infinite support.} Now let \(G\) be the c.d.f.\ of the zeta distribution~\eqref{zetadistr} with parameter \(s=5\). As was discussed in Section~\ref{secH}, in this case \(|\M[G](\u+\i v)|\) is strictly positive for any 
\[\u<-\frac{\log(\zeta(5))}{\log (2)}+1\approx 0.95.\]
\item 
\underline{Continuous mixing distribution.} Finally, assume that \(G\) corresponds to the c.d.f.\ of the uniform distribution on \([0,1]\). In this case, we have \[|\M[G](\u+\i v)|=|\u+\i v|^{-1}>0\]for any \(\u, v\in\R\).
\end{enumerate}

\begin{figure}[t]
\includegraphics[width=1\linewidth]{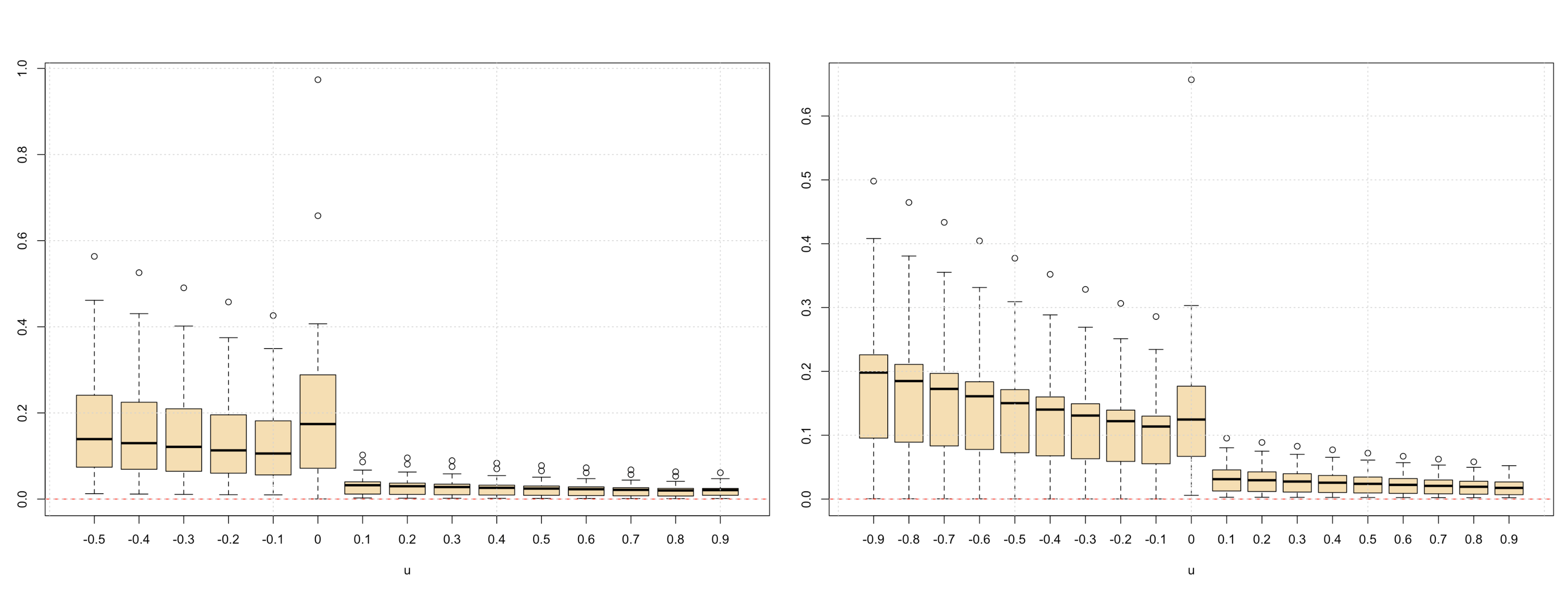}
	\caption{The estimates \(\widehat{\mathcal{R}}(\widehat{F};\u,x)\) for different values of \(\u\) for Beta(2,2) (left) and Gamma(2,2) (right) distributions for \(G\) being a c.d.f.\ of a discrete distribution supported on two points}
	\label{fig:u_twopoint}
\end{figure}
\begin{figure}[t]
\includegraphics[width=1\linewidth]{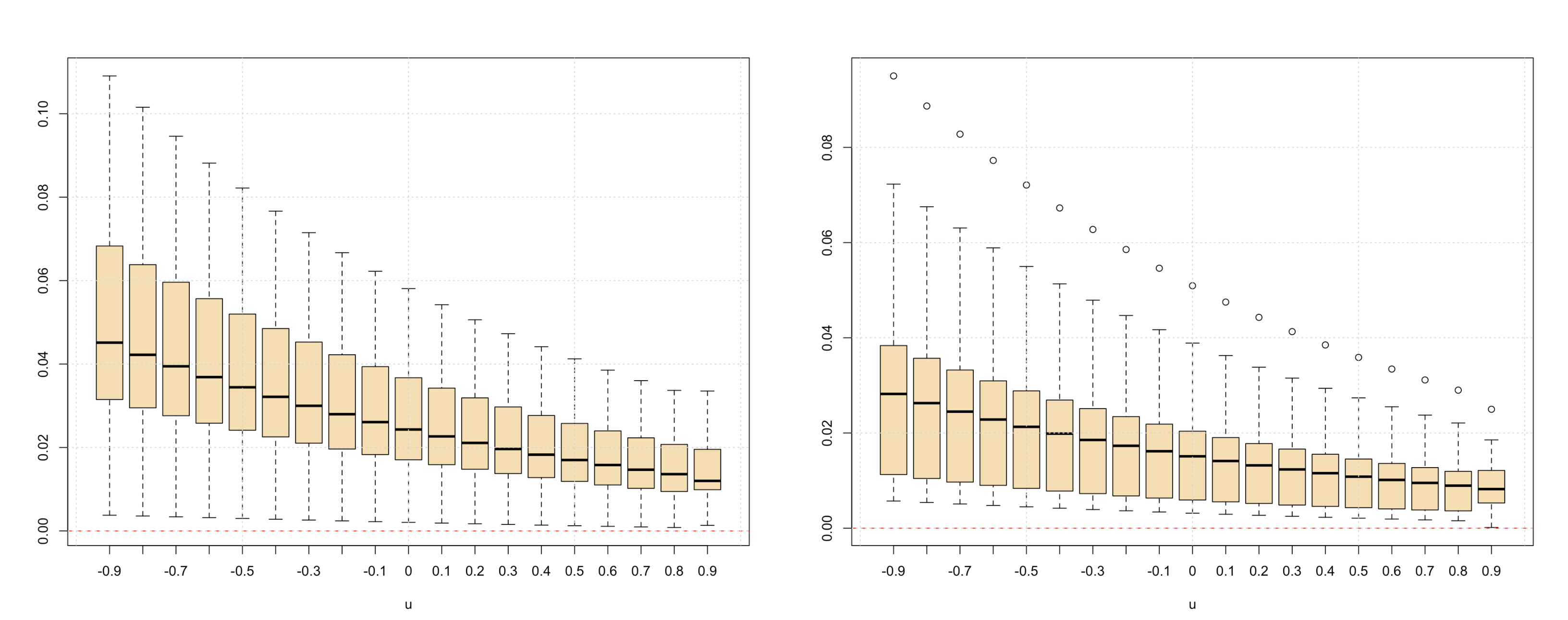}
	\caption{The estimates \(\widehat{\mathcal{R}}(\widehat{F};\u,x)\) for different values of  \(\u\) for Beta(2,2) (left) and Gamma(2,2) (right) distributions for \(G\) being a c.d.f.\ of the zeta distribution}
	\label{fig:u_zeta}
\end{figure}
\begin{figure}[t]
\includegraphics[width=1\linewidth]{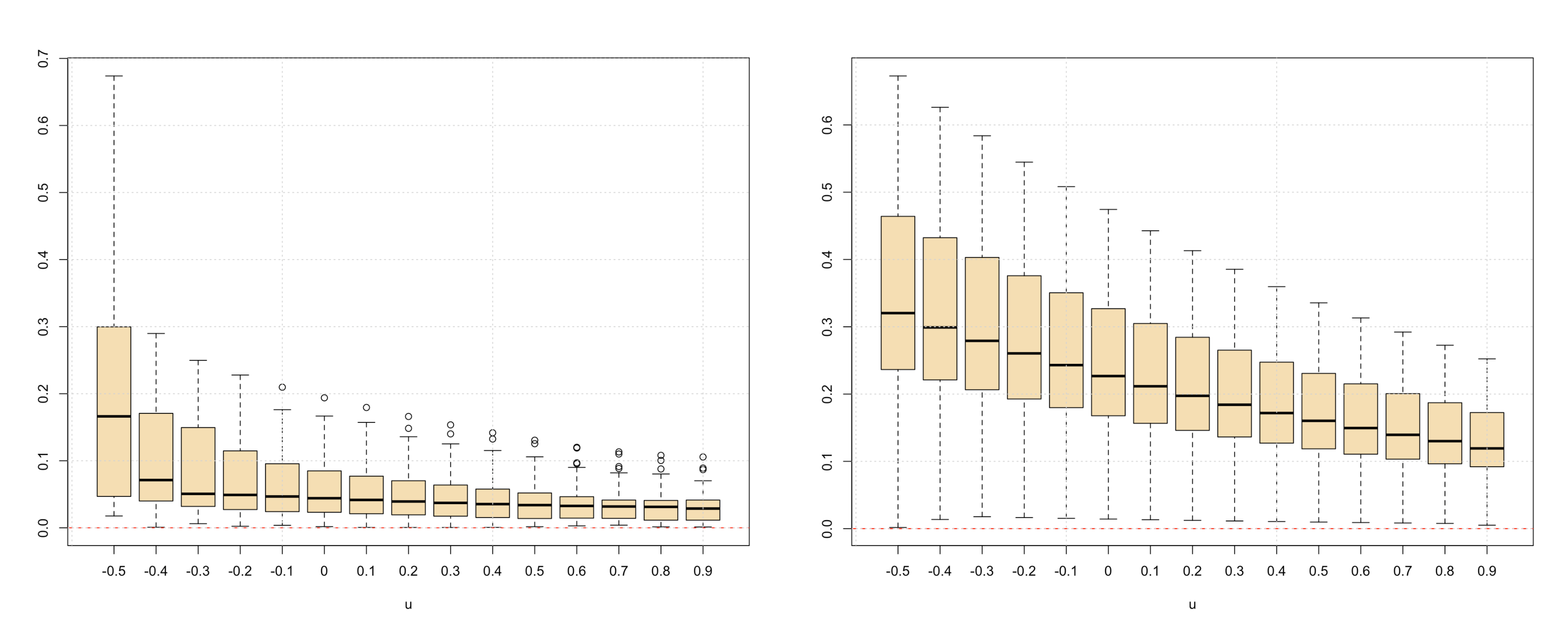}
	\caption{The estimates \(\widehat{\mathcal{R}}(\widehat{F};\u,x)\) for different values of \(\u\) for Beta(2,2) (left) and Gamma(2,2) (right) distributions for \(G\) being a c.d.f.\ of the uniform distribution}
	\label{fig:u_unif}
\end{figure}
As follows from the definition~\eqref{est_main} of the proposed estimator and the form~\eqref{est_main2} of the kernel, the choice of the parameters \(\u\) and \(T\) is essential. Figures~\ref{fig:u_twopoint}, \ref{fig:u_zeta} and~\ref{fig:u_unif} demonstrate the dependence of the left-hand side in Theorem~\ref{thm1} on \(\u\) taken from \(-0.9\) to \(0.9\) with a step of \(0.1\). In case of the two-point distribution \(G\) with \(F\) being the beta law, as well as that of the uniform mixing distribution, the boxplots are displayed for \(u\) starting from \(-0.5\), due to the higher error values obtained for \(\u\in[-0,9,-0.5)\) that disrupt the picture. The value \(x\) is fixed as \(x=0.5\), and the values of \(\widehat{\mathcal{R}}\left(\widehat{F};\u,x\right)\) were computed based on  \(25\) samples of size \(n=1000\). The plots indicate that, as suggested by Theorem~\ref{thm1}, the proposed estimator becomes more accurate as \(\u\) increases with an evident exception for \(\u=0\) in case of the discrete mixing distribution supported on two points. As was mentioned earlier in this section, the reason for this peculiarity is that \(\u=0\) turns out to be exactly the threshold value~\eqref{u_circ}, for which the Mellin transform \(\M[G](z)\) of the mixing distribution is not guaranteed to be separated from zero in absolute value.

In what follows, the parameter \(\u\) is fixed as \(\u=0.5.\)
As for the parameter \(T\), based on the discussion in Section~\ref{convrates} it should be chosen of order \(T\gtrsim \sqrt{n}\). While the choice of \(T\) such that \(T / \sqrt{n} \to \infty\) will not influence the overall asymptotics of the bound, the greater values of \(T\) allow to obtain smaller values of error in numerical studies. To this end, in case of the discrete mixing distributions we choose \(T=n\) as the optimal value allowing to achieve the smallest error possible while not sacrificing too much in terms of the computational cost. In case of the continuous mixing distribution, however, since \(|\M[G](\u+\i v)|\) decreases in \(v\) and can become arbitrarily close to zero as \(|v|\to \infty\), the choice of the truncation level \(T\) for the estimator~\eqref{est_main} should be approached more carefully. For both the beta and the gamma distributions we choose \(T\) as the value minimising the corresponding mean-squared error of the estimator based on 500 samples of size \(n=1000\), resulting into the values \(T=34.6\) and \(T=29.7\) for each of the laws, respectively. For the estimator~\eqref{fourier}, the parameter \(R\) is chosen by this procedure for all the three cases of mixing distributions \(G\), leading to the values \(3.5\), \(9.6\) and \(9.7\) for \(G\) being the two point, zeta and uniform distribution, respectively, with \(F\) being the beta law, and \(9.7\), \(45.4\) and \(3.4\) for the respective mixing distributions \(G\) with \(F\) being the gamma law.

\begin{table}[h]
\centering
\caption{Average error~\eqref{error_def} for the Mellin-based estimator~\eqref{est_main} and the Fourier-based estimator ~\eqref{fourier} obtained via 100 simulation runs}
\label{mses}
\subcaption*{Beta(2,2) distribution}
\begin{tabular}{cc|cc|c}
\multicolumn{2}{c|}{\multirow{2}{*}{Mixing distribution}} & \multicolumn{2}{c|}{Method} & Relative error  \\
&& Mellin & Fourier & Fourier/Mellin \\ \hline
\multirow{3}{*}{Two point} & \(n=100\) & 0.0012 & 0.2128 & 177.33 \\
& \(n=500\) & 0.0003 & 0.2093 & 697.67 \\
& \(n=1000\) & 0.0002 & 0.2092 & 1046 \\ \hline
\multirow{3}{*}{Zeta} & \(n=100\) & 0.000069 & 0.00011 & 1.59 \\
& \(n=500\) & 0.000074  & 0.00011 & 1.49  \\
& \(n=1000\) & 0.000068 & 0.0001 & 1.47 \\ \hline
\multirow{3}{*}{Uniform} & \(n=100\) & 0.0051 & 0.0046 & 0.9 \\
& \(n=500\) & 0.0014 & 0.0012 & 0.86 \\
& \(n=1000\) & 0.00071 & 0.00072 & 1.01 \\ \hline
\end{tabular}
\subcaption*{Gamma(2,2) distribution}
\begin{tabular}{cc|cc|c}
\multicolumn{2}{c|}{\multirow{2}{*}{Mixing distribution}} & \multicolumn{2}{c|}{Method} & Relative error \\
&& Mellin & Fourier & Fourier/Mellin \\ \hline
\multirow{3}{*}{Two point} & \(n=100\) & 0.0012 & 0.0728 & 60.67 \\
& \(n=500\) & 0.0003 & 0.0723 & 241 \\
& \(n=1000\) & 0.0001 & 0.0716 & 716 \\ \hline
\multirow{3}{*}{Zeta} & \(n=100\) & 0.0008 & 0.0022 & 2.75 \\
& \(n=500\) & 0.0002 & 0.0004 & 2 \\
& \(n=1000\) & 0.00008 & 0.0002 & 2.5 \\ \hline
\multirow{3}{*}{Uniform} & \(n=100\) & 0.004 & 0.0035 & 0.88 \\
& \(n=500\) & 0.001 & 0.0007 & 0.7 \\
& \(n=1000\) & 0.0006 & 0.0004 & 0.67 \\ \hline
\end{tabular}
\end{table}

Table~\ref{mses} represents the mean-squared errors
\begin{equation} \label{error_def}
err_n=\frac{1}{KJ}\sum\limits_{k=1} ^K \sum\limits_{j=1} ^J \left(\widehat{F}^{(k)} _n (x_j) - F(x_j)\right)^2,
\end{equation}
where \(\widehat{F}^{(k)} _n\) is either the Mellin-based estimate~\eqref{est_main} or the Fourier-based estimate~\eqref{fourier} obtained on a \(k\)-th sample of size \(n\), for different values of \(n\in\{100,500,1000\}\) and \(K=100\) simulation runs. The values \(x_j\) are taken on a \(J\)-point equidistant grid \(\{x_j\}_{j\in\{1,\dots,J\}}\), the latter being over [0.1,2] with a step of 0.1 in case of the beta distribution and over [0.01,5] with the same step in case of the gamma law.

It can be seen that for both beta and gamma distributions and all the considered mixing laws the proposed estimator~\eqref{est_main} converges to the true c.d.f.\ \(F\) rather fast, with the error~\eqref{error_def} decreasing as the sample size grows and being reasonably small already for \(n=100\). In case of the discrete mixing distribution supported on two points, the estimator~\eqref{est_main} based on the Mellin transform significantly outperforms the Fourier-based estimator~\eqref{fourier}, leading to the error values that are more than 60 times smaller than those of the latter. In case of the zeta distribution, the proposed estimator~\eqref{est_main} also demonstrates a better performance, although this difference is less striking when \(F\) is the c.d.f.\ of the beta law, since the average errors for both estimators are smaller than \(0.00008\) already for \(n=100\), and further convergence to zero appears to be very slow. Finally, in case of the uniform mixing distribution both estimators~\eqref{est_main} and~\eqref{fourier} provide very similar results, with the latter estimator providing a bit better fit in case of the gamma distribution. Also, while for both the beta and the gamma distributions the error decreases as the sample size grows, the rate of convergence to zero appears to be slower than for discrete mixing distributions considered above.

\appendix
\section{Berry -- Esseen inequality in a particular case}\label{exBE}
Let us consider the case when \(\varphi\) and \(\psi\) are the distribution functions of the exponential distribution, that is, 
\begin{eqnarray*}
\varphi(x) =  \bigl( 1- e^{-\lambda_1 x} \bigr) \I\bigl\{ x \geq 0 \bigr\}, \qquad 
\psi(x) =  \bigl( 1- e^{-\lambda_2 x} \bigr) \I\bigl\{ x \geq 0 \bigr\},
\end{eqnarray*} 
where \(\lambda_1, \lambda_2 >0.\) Since
\begin{eqnarray*}
\M[\varphi](z) = \lambda_1^{1-z} \Gamma(z), \qquad 
\M[\psi](z) = \lambda_2^{1-z} \Gamma(z)
\end{eqnarray*}
for any \(z\) with \(\Re(z) >0,\)
we can take \(\u=1/2\).  We have 
 \begin{eqnarray*}
\rho_{1/2}(\F,\G):= \sup\limits_{x \geq 0} \Bigl|\frac{ e^{-\lambda_1 x}-e^{-\lambda_2 x}}{\sqrt{x}} \Bigr|.
\end{eqnarray*}
The maximum is attained at some point  \(x_0\in (0,\infty).\)  For this numerical example we take the values \(\lambda_1 = 1, \lambda_2 =1.5,\) for which \(x_0=0.4\).

We also take \(b=0.8\), and numerically get \(c(0.8) \approx 4.8\) as the solution of~\eqref{cb}. We aim to study the right-hand side of~\eqref{berry}  depending on \(T.\) 

Figure~\ref{plot1} represents the plots of the first and the second summands in~\eqref{berry}, namely,
\begin{eqnarray*}
\widetilde{I}_1 &=& 
\frac{b}{2}\int\limits_{-T} ^T \frac{|\lambda_1^{1/2 - \i v} - \lambda_2^{1/2 - \i v}|}{|v|} \bigl|\Gamma((1/2)+\i v)\bigr|\,dv,\\
\widetilde{I}_2 &=& bT x_0^{\u-1} \int_0^{2c(b)/T}\bigl(e^{x_0e^r}- e^{x_0}\bigr)\,dr.
\end{eqnarray*}
An interesting observation is that \(\widetilde{I}_1\) is almost equal to a constant for large \(T\), while \(\widetilde{I}_2\) decays rather fast as \(T\) grows.
\begin{figure}[t]
\centering
\includegraphics[width=0.8\linewidth]{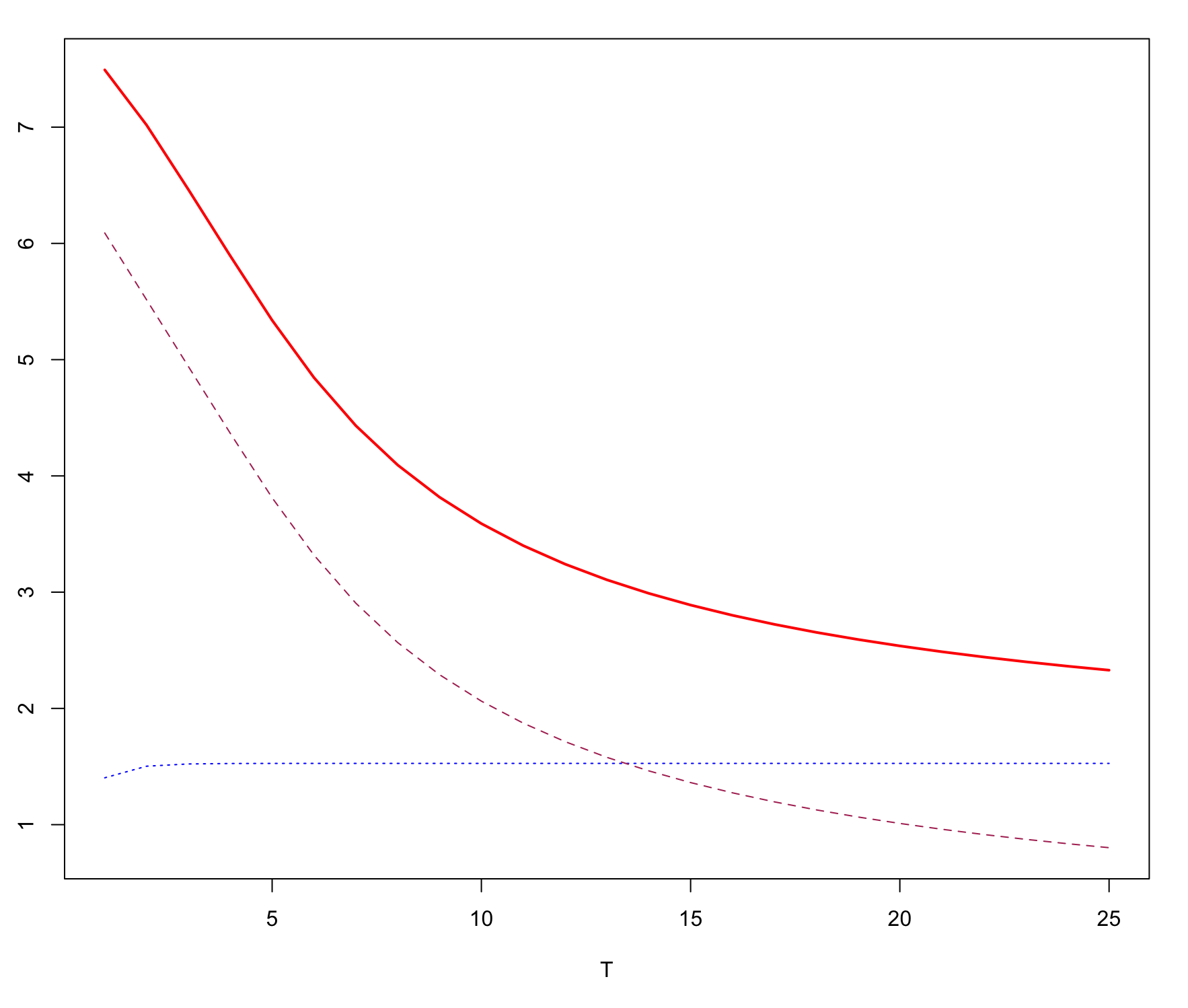}
	\caption{The graphs of the first and the second summands, and the right-hand side in~\eqref{berry}, depending on the value of \(T.\) The blue dotted line  represents the first summand and is for large \(T\) is close to 1.1. The maroon dashed line shows the behaviour of the second summand, which tends to 0 as \(T \to \infty.\) Finally, the red solid line shows the total behaviour of the r.h.s. in \eqref{berry}.
	}
	\label{plot1}
\end{figure}

\section{Proofs}\label{AppB}
\subsection{Proof of Lemma~\ref{lemm}}
\label{A1}
Let us rewrite~\eqref{def_mellin} as
\begin{multline*}
\M[\F](z) =  \int\limits_0 ^{\infty} x^{z-1}\,d\F(x) = \int\limits_0 ^{\infty} e^{(z-1)\log x} \,d\F(e^{\log x})\\
= \int\limits_{-\infty} ^{\infty} e^{-sy} \,d\F(e^y) =: \mathcal{L}[\F(e^y)](s), 
\end{multline*}
where the right-hand side defines the bilateral Laplace transform of \(\F(e^y)\), \(y=\log x\), at the point \(s=-(z-1)\). Applying Theorem 7.7.5 from~\cite{Kawata} for the bilaterial Laplace transform, we arrive at \begin{eqnarray*}
\frac{1}{2}\Bigl( \F(e^y+0) + \F(e^y+0) 
\Bigr) - \F(0)&=& \frac{1}{2\pi \i} \int\limits_{\tilde{u}-\i\infty} ^{\tilde{u}+\i\infty} e^{sy}\frac{\mathcal{L}[\F(e^y)](s)}{s}\,ds \\ &=& \frac{1}{2\pi \i} \int\limits_{\u-\i\infty} ^{\u+\i\infty} x^{-z+1}\frac{\M[\F](z)}{-(z-1)}\,dz,
\end{eqnarray*}
provided \(\tilde{u}:=1-\u > 0\). The second statement of the theorem follows from  another part of Theorem 7.7.5 from~\cite{Kawata}

\subsection{Proof of Lemma~\ref{lem1}}
\label{A}

Let \begin{eqnarray}\label{wW}
w(x) = \frac{\sin^2 (T\log x)}{\pi Tx^{\u}\log^2 x} \qquad  \mbox{and} \qquad W(x) = \int\limits_{0} ^x w(y)\,dy, x \in \R_+. 
\end{eqnarray} Define 
\[\widetilde{\F}(x) := \int\limits_0 ^{\infty} \F(x/y)\,dW(y), \qquad \widetilde{\G}(x) := \int\limits_0 ^{\infty} \G(x/y)dW(y).\]
The main idea of the proof is to find upper and lower bounds for \(\rrho(\widetilde{\F}, \widetilde{\G})\) such that the lower bound linearly depends on \(\rrho(\F, \G)\) with positive slope.

\textbf{1. Upper bound for  \(\rrho(\widetilde{\F}, \widetilde{\G})\).}
 Due to Theorem~7.8.3 from \cite{Kawata}, 
\[\M[\widetilde{\F}](z) = \M[\F](z) \M[W](z),
\qquad \M[\widetilde{\G}](z) = \M[\G](z) \M[W](z),\]
and, by the inverse formula for the Mellin transform (Lemma~\ref{lemm}), 
\begin{multline*}
\widetilde{\F}(x)-\widetilde{\G}(x) 
%=\frac{1}{2\pi \i} \int\limits_{\u-\i\infty}^{\u+\i\infty} \frac{\M[\F](z)-\M[\G](z)}{-(z-1)}\M[W](z)x^{-z+1}\,dz 
=  \frac{1}{2\pi} \int\limits_{-\infty} ^{\infty} \frac{\M[\F](\u+\i v)-\M[\G](\u+\i v)}{-(\u+\i v-1)} \\
\hspace{4cm} \times\M[W](\u+\i v)x^{-\u-\i v+1}\,dv.
\end{multline*}
As
\begin{eqnarray}
\nonumber \M[W](\u+\i v) &=& \int\limits_0 ^{\infty} x^{(\u-1)+\i v} \frac{\sin^2 (T\log x)}{\pi Tx^{\u}\log^2 x}\,dx \nonumber\\ &=& \int\limits_{-\infty}^{\infty} e^{\i vy} \frac{\sin^2 (Ty)}{\pi Ty^2}\,dy  = \left(1-\left|\frac{v}{T}\right|\right)\I\{|v|\leq T\},
\label{MW}
\end{eqnarray}
we get
\begin{equation}
\rrho(\widetilde{\F}, \widetilde{\G})
%&\leq& \frac{1}{2\pi} \int\limits_{-T} ^T \frac{|\M[\F](\u+\i v) - \M[\G](\u+\i v)|}{\sqrt{(\u-1)^2+v^2}}\,dv\\
\leq\frac{1}{2\pi}\int\limits_{-T} ^T \frac{|\M[\F](\u+\i v) - \M[\G](\u+\i v)|}{|v|}\,dv.\label{bound1}
\end{equation}
\textbf{2. Lower bound for  \(\rrho(\widetilde{\F}, \widetilde{\G})\).} We have
\begin{equation*}\widetilde{\F}(x)-\widetilde{\G}(x)  = \int\limits_0 ^{\infty} \Bigl(\F\left(x/y\right) - \G\left(x/y\right)\Bigr) \frac{\sin^2 (T\log y)}{\pi Ty^{\u}\log^2 y}\,dy.\end{equation*}
Let us denote \(\Delta=\rho_{\u}(\F,\G)\). Then  \[x_0^{\u-1}\bigl(\F(x_0)-\G(x_0\pm 0)\bigr)=\Delta\] or \[x_0^{\u-1}\bigl(\G(x_0 \pm 0)-\F(x_0-0)\bigr)=\Delta.\] Since for both cases the proof follows the same lines, we will only consider the first one. In this case
\begin{multline*}
x_0^{\u-1} e^{\frac{c(b)(\u-1)}{T}}
\Bigl( 
\widetilde{\F}\left(x_0 e^{\frac{c(b)}{T}}\right) - \widetilde{\G}\left(x_0 e^{\frac{c(b)}{T}}\right)
\Bigr) \\ 
%= 
%x_0^{\u-1} e^{\frac{c(b)(\u-1)}{T}}
%\int\limits_0 ^{\infty} \Bigl(\F\bigl(x_0 e^{\frac{c(b)}{T}}/y\bigr)-\G\bigl(x_0 e^{\frac{c(b)}{T}}/y\bigr)\Bigr)\frac{\sin^2(T\log y)}{\pi Ty^{\u}\log^2 y}\,dy \\ 
%= x_0^{\u-1} \int\limits_{-\infty} ^{\infty} \Bigl(\F\bigl(x_0 e^{\frac{c(b)-r}{T}}\bigr) - \G\bigl(x_0 e^{\frac{c(b)-r}{T}}\bigr)\Bigr)\frac{\sin^2 r}{\pi r^2}\,e^{(c(b)-r)(\u-1)/T}dr \\ = x_0^{\u-1} \int\limits_{|r|<c(b)} \Bigl(\F\bigl(x_0 e^{\frac{c(b)-r}{T}}\bigr) - \G\bigl(x_0 e^{\frac{c(b)-r}{T}}\bigr)\Bigr)\frac{\sin^2 r}{\pi r^2}e^{(c(b)-r)(\u-1)/T}\,dr \\  \hspace{1cm}
%+ x_0^{\u-1}  \int\limits_{|r|\geq c(b)} \Bigl(\F\bigl(x_0 e^{\frac{c(b)-r}{T}}\bigr) - \G\bigl(x_0 e^{\frac{c(b)-r}{T}}\bigr)\Bigr)\frac{\sin^2 r}{\pi r^2}e^{(c(b)-r)(\u-1)/T}\,dr \\ 
\geq  x_0^{\u-1} \int\limits_{|r|<c(b)} \Bigl(\G\bigl(x_0 \pm 0\bigr) - \G\bigl(x_0 e^{\frac{c(b)-r}{T}}\bigr)\Bigr)\frac{\sin^2 r}{\pi r^2}e^{(c(b)-r)(\u-1)/T}\,dr \\  \hspace{1cm}+  x_0^{\u-1}  \int\limits_{|r|<c(b)} \Bigl(\F\bigl(x_0 e^{\frac{c(b)-r}{T}}\bigr) - \G\bigl(x_0 \pm 0\bigr)\Bigr)\frac{\sin^2 r}{\pi r^2} e^{(c(b)-r)(\u-1)/T}\,dr \\  \hspace{5cm} - \Delta\int\limits_{|r|\geq c(b)} \frac{\sin^2 r}{\pi r^2} dr.
\end{multline*}
Since for any \(r<c(b)\) it holds 
\[
x_0^{\u-1} 
%e^{(c(b)-r)(\u-1)/T}
\Bigl(
\F\bigl(x_0 e^{\frac{c(b)-r}{T}}\bigr) - \G
\bigl(x_0 \pm 0\bigr)
\Bigr)
\geq 
x_0^{\u-1}\bigl(
\F\left(x_0 \right) - \G
\left(x_0 \pm 0\right)
\bigr)=\Delta,\]
 we further have
\begin{multline*}
x_0^{\u-1} e^{\frac{c(b)(\u-1)}{T}}
\Bigl(
\widetilde{\F}\bigl(x_0 e^{\frac{c(b)}{T}}\bigr) - \widetilde{\G}\bigl(x_0 e^{\frac{c(b)}{T}}\bigr)
\Bigr) 
% \\ \geq -
% \frac{x_0^{\u-1} }{\pi}
% \int\limits_{|r|<c(b)}  \left|\G(x_0 \pm 0) - \G\left(x_0 e^{\frac{c(b)-r}{T}}\right)\right| e^{(c(b)-r)(\u-1)/T}\,dr 
% \\ \hspace{1cm} 
%+ \Delta\int\limits_{|r|<c(b)} \frac{\sin^2 r}{\pi r^2}e^{(c(b)-r)(\u-1)/T}\,dr - \Delta \int\limits_{|r|\geq c(b)} \frac{\sin^2 r}{\pi r^2}\,dr 
\\ 
\hspace{-3cm}
\geq -\frac{T}{\pi}x_0^{\u-1}\int\limits_0 ^{2c(b)/T} |\G(x_0 \pm 0) - \G(x_0 e^r)|\,dr 
\\ 
+ \Delta\int\limits_{|r|<c(b)} \frac{\sin^2 r}{\pi r^2}e^{(c(b)-r)(\u-1)/T}\,dr - \Delta \int\limits_{|r|\geq c(b)} \frac{\sin^2 r}{\pi r^2}\,dr.\hspace{0.6cm}
\end{multline*}
%Now, since \(\u<1\), we get that
%\begin{eqnarray*}
%x_0^{\u-1} e^{\frac{c(b)(\u-1)}{T}}
%\Bigl(
%\widetilde{F}\bigl(x_0 e^{\frac{c(b)}{T}}\bigr) - \widetilde{\G}\bigl(x_0 e^{\frac{c(b)}{T}}\bigr) 
%\Bigr)
% &\geq& -\frac{T}{\pi}
%x_0^{\u-1}\int\limits_0 ^{2c(b)/T} |\G(x_0 \pm 0) - \G(x_0 e^r)|\,dr \\ && \hspace{0.6cm}+ \Delta 
%\left( \int\limits_{|r|<c(b)} \frac{\sin^2 r}{\pi r^2} e^{(c(b)-r)(\u-1)/T}dr - \int\limits_{|r|\geq c(b)} \frac{\sin^2 r}{\pi r^2}\;dr \right).
%\end{eqnarray*}
Now, noting that 
\begin{eqnarray*}
\int\limits_{|r|<c(b)} \frac{\sin^2 r}{\pi r^2}e^{(c(b)-r)(\u-1)/T}dr 
&\geq& 
e^{2c(b)(\u-1)/T}
\int\limits_{|r|<c(b)} \frac{\sin^2 r}{\pi r^2} dr,\\
\int\limits_{|r|\geq c(b)} \frac{\sin^2 r}{\pi r^2}\;dr
&=& 1 - 
\int\limits_{|r|<c(b)} \frac{\sin^2 r}{\pi r^2}
dr
%\int\limits_{|r|<c(b)} \frac{\sin^2 r}{\pi r^2}e^{r(1-\u)/T}\,dr+
%\int\limits_{|r|\geq c(b)} \frac{\sin^2 r}{\pi r^2}e^{r(1-\u)/T}\,dr \geq \int_\R  \frac{\sin^2 r}{\pi r^2} dr  = 1,
\end{eqnarray*}
we obtain \begin{multline*} \rrho (\widetilde{\F}, \widetilde{\G}) \geq -\frac{T}{\pi}
x_0^{\u-1}\int\limits_0 ^{2c(b)/T} |\G(x_0 \pm 0) - \G(x_0 e^r)|\,dr \\+\Delta\Bigl(-1+\int\limits_{|r|< c(b)} \frac{\sin^2 r}{\pi r^2}\,dr \cdot
\bigl(
e^{2c(b)(\u-1)/T} +1
\bigr)
\Bigr),%\\
%&\geq& -\frac{T}{\pi}e^{c(b)(1-\u)/T}\int_0^{2c(b)/T} \left|\G(x_0 \pm 0)-\G\left(x_0 e^r\right)\right|\,dr \\
%&&\hspace{3cm}+ \Delta\left(1-2\int\limits_{|r|>c(b)} \frac{\sin^2 r}{\pi r^2}\,dr \right),%\\
%&\geq&
%\textcolor{red}{-\frac{T}{\pi}e^{c(b)(1-\u)/T}\int\limits_{|r|<\frac{2c(b)}{T}} \left|\G(x_0 \pm 0)-\G\left(x_0 e^r\right)\right|\,dr }%\\
%&&\hspace{3cm} \textcolor{red}{+ \Delta\left(2\int\limits_{|r|<c(b)} \frac{\sin^2 r}{\pi r^2}\,dr - 1\right)}
%,
 \end{multline*}
 where \(e^{2c(b)(\u-1)/T} +1 > 3/2\) due to our choice of \(T\). Choosing \(c(b)\) as in~\eqref{cb}, arrive at the upper bound for \(\Delta\)
\begin{equation}
\label{bound2}
\Delta \leq 
\pi b \cdot
\rrho (\widetilde{\F}, \widetilde{\G}) + T b
x_0^{\u-1}\int\limits_0 ^{2c(b)/T} |\G(x_0 \pm 0) - \G(x_0 e^r)|\,dr.
\end{equation}
Combining this bound with~\eqref{bound1}, we obtain the result.
\subsection{Proof of Theorem~\ref{thm1}}
\label{B}Let us consider 
\begin{equation}
\label{triang}
x^{\u-1}|F(x)-\widehat{F}(x)| \leq x^{\u-1} |F(x)-\widetilde{F}(x)| + x^{\u-1} |\widetilde{F}(x)-\widehat{F}(x)|
\end{equation}
with \(\widetilde{F}(x)\) being an auxiliary function defined as the multiplicative convolution \[\widetilde{F}(x) := F\star W(x) = \int\limits_0 ^{\infty} F\left(x/y\right)\,dW(y),\] where \(W: \R_+ \to \R_+\) is defined by~\eqref{wW}. Since \(\M[\widetilde{F}](z) = \M[F](z)\M[W](z)\) for any \(z\) such that \(\Re(z) \in \CC_{\widetilde{F}}=\CC_F \cap \CC_W\), and \(u^\circ \in \CC_W\) by \eqref{MW}, we get from Lemma~\ref{lemm},
\[\widetilde{F}(x) = \frac{1}{2\pi} \int\limits_{-\infty} ^{\infty} x^{-\u-\i v+1}\frac{\M[F](\u+\i v)\M[W](\u+\i v)}{-(\u+\i v-1)}\,dv, \quad v\in\R.\] Noticing that
\(\M[W](\u+\i v) = K(v),\) see \eqref{est_main2} and \eqref{MW}, we have for the second summand in~\eqref{triang} 
\begin{multline*}
x^{\u-1}|\widetilde{F}(x)-\widehat{F}(x)| 
\\
\hspace{-0.7cm}= x^{\u-1} \left|\frac{1}{2\pi} \int\limits_{-\infty} ^{\infty} x^{(-\u+1)-\i v} \frac{\M[F](\u+\i v)-\widehat{\M[F]}(\u+\i v)}{-(\u+\i v-1)} K(v) \right| \\ 
= x^{\u-1} \left|\frac{1}{2\pi n}\sum\limits_{k=1} ^n \int\limits_{-\infty} ^{\infty}x^{(-\u+1)-\i v} \frac{\M[F_{\text{mix}}](\u+\i v)-X_k ^{(\u-1)+\i v}}{-(\u+\i v-1)\M[G](\u+\i v)}K(v) \right| \\ = \frac{x^{\u-1}}{2\pi n} \left|\sum\limits_{k=1} ^n \Lambda_k (X_k, x)\right|.\hspace{7.2cm}
\end{multline*}
As for the first summand in~\eqref{triang}, let us note that
\begin{multline}
\sup_{x \geq 0} \Bigl| x^{\u-1}\bigl( F(x)-\widetilde{F}(x) \bigr)
\Bigr|
=
\sup_{x \geq0} \Biggl| 	x^{\u-1} F(x) - \int\limits_0 ^{\infty}  \left(
x/y
\right)^{\u-1}F\left(x/y\right)
\widetilde{w}(y) dy
\Biggr| \\
=
\sup_{x \geq 0} \Biggl| 
\int\limits_0 ^{\infty}	 
\left( x^{\u-1} F\bigl(x\bigr) -   \left(
x/y
\right)^{\u-1}F\left(x/y\right)
\right) \widetilde{w}(y) dy\Biggr|,
 \label{e11}
\end{multline}
where \(\widetilde{w}(y):=w(y)y^{\u-1}, y \in \R_+\) is a density function of some distribution.
Since we assume that the distribution function \(F\) is continuous at least in a small vicinity of \(0\), and \eqref{add} holds, the supremum in~\eqref{e11} is attained at some point \(x_0>0.\) Thus, by Lemma~\ref{lem1},
\begin{multline*}
x^{\u-1}|F(x)-\widetilde{F}(x)| \leq
 \frac{b}{2}\int\limits_{-T} ^{T}  \frac{\left| \M[F](\u+\i v)-\M[\widetilde{F}](\u+\i v)\right|}{| v |}\,dv \\+ bT x_0^{\u-1}  \int_0^{2 c(b)/ T} \left|F(x_0 \pm 0)-F\left(x_0 e^r\right)\right|\,dr.
\end{multline*}
Now it only remains to observe that
\begin{multline*}
\int\limits_{-T} ^{T}  \frac{\left| \M[F](\u+\i v)-\M[\widetilde{F}](\u+\i v)\right|}{| v |}\,dv \\=
\int\limits_{-T} ^{T}  \frac{\left| \M[F](\u+\i v)(1-\M[W](\u+\i v))\right|}{| v |}\,dv
 =\frac{1}{T} \int\limits_{-T} ^{T}  \left| \M[F](\u+\i v)\right|\,dv.
\end{multline*}
\subsection{Proof of Theorem~\ref{thm2}}
\label{C}

By the upper bound established in Theorem~\ref{thm1},
\begin{eqnarray}
\nonumber
\RR^*(\widehat{F}; \u, x)  &\leq&    \frac{3b^2}{4T^2}\Bigl(\int\limits_{-T} ^T |\M[F](\u+\i v)|\,dv\Bigr)^2  \\
\nonumber && \hspace{1cm} + 3 b^2 T^2 x_0^{2(\u-1)} \Bigl(\int\limits_0 ^{2c(b)/T} |F(x_0 \pm 0)-F(x_0 e^r)|\,dr\Bigr)^2  \\
&&
 \hspace{1cm}+  \frac{3x^{2(\u-1)}}{4\pi^2 n^2}\E\left[\left|\sum\limits_{k=1} ^n \Lambda_k (X_k, x)\right|^2\right].
 \label{bbr}
\end{eqnarray}
%\begin{multline*}
%\E\left[\left(F(x)-\widehat{F}(x)\right)^2\right] \leq 2\left(\frac{b^2}{16\pi^2} \left|\frac{1}{T}\int\limits_{-T} ^T e^{-iv\log x} \frac{\M[F](1+iv)|v|}{-(1+iv)}\,dv + \int\limits_{|v|>T} e^{-iv\log x} \frac{\M[F](1+iv)}{-(1+iv)}\,dv\right|^2\right. \\ \hspace{4cm} + b^2 T^2 \left(\int\limits_{|r|<\frac{2c(b)}{T}} \left|F(x)-F\left(xe^r\right)\right|\,dr\right)^2 \\ + \frac{1}{4\pi^2 n^2} \E\left[\left|\sum\limits_{k=1} ^n \Lambda_k (X_k, x)\right|^2\right]\Biggl.\Biggr),
%\end{multline*}
%and similarly, 
%\begin{multline*}
% \left|\frac{1}{T}\int\limits_{-T} ^T e^{-iv\log x} \frac{\M[F](1+iv)|v|}{-(1+iv)}\,dv + \int\limits_{|v|>T} e^{-iv\log x} \frac{\M[F](1+iv)}{-(1+iv)}\,dv\right|^2 \\ \leq 2\left( \left|\frac{1}{T}\int\limits_{-T} ^T e^{-iv\log x} \frac{\M[F](1+iv)|v|}{-(1+iv)}\,dv\right|^2 + \left|\int\limits_{|v|>T} e^{-iv\log x} \frac{\M[F](1+iv)}{-(1+iv)}\,dv\right|^2\right).
%\end{multline*}
Now, since \(\Lambda_k (X_k, x)\) are centred i.i.d.\ for all \(k=1,\dots,n\),
\begin{eqnarray*}
\E\left[\left|\sum\limits_{k=1} ^n \Lambda_k (X_k, x)\right|^2\right] 
=\Var\left( \sum\limits_{k=1} ^n
 \Lambda_k (X_k, x)
\right)=
n\E\left[\Lambda_1 (X_1, x)\overline{\Lambda_1 (X_1, x)}\right],
\end{eqnarray*}
where for the latter quantity we have that
\begin{multline*}
\E\left[\Lambda_1 (X_1, x)\overline{\Lambda_1 (X_1, x)}\right] \\
= \E\left[\int\limits_{-T} ^T \int\limits_{-T} ^T x^{2(-\u+1)-\i(v-w)}\frac{(\M[F_{\text{mix}}](\u+\i v)-X_1 ^{(\u-1)+\i v})}{(\u-1+\i v)(\u-1-\i w)}\left(1-\frac{|v|}{T}\right)\right. 
\\ \hspace{3.5cm} \times \left. \frac{(\overline{\M[F_{\text{mix}}](\u+\i w)}-X_1 ^{(\u-1)-\i w})}{\M[G](\u+\i v)\overline{\M[G](\u+\i w)}} \left(1-\frac{|w|}{T}\right)\,dv\,dw\right] \\
\leq \int\limits_{-T} ^T \int\limits_{-T} ^T \frac{\M[F_{\text{mix}}](2\u-1+\i(v-w))}{(\u-1+\i v)(\u-1-\i w)\M[G](\u+\i v)\overline{\M[G](\u+\i w)}}\\ 
\hspace{4.5cm} \times x^{2(-\u+1)-\i(v-w)}\left(1-\frac{|v|}{T}\right)\left(1-\frac{|w|}{T}\right)\,dv\,dw.
\end{multline*}
%Since
%\begin{multline*}
%\E[(\M[F_{\text{mix}}](\u+\i v)-X_1 ^{(\u-1)+\i v})(\overline{\M[F_{\text{mix}}](\u+\i w)}-X_1 ^{(\u-1)-\i w})] \\ 
% = 
%-\M[F_{\text{mix}}](\u+\i v)\M[F_{\text{mix}}](\u-\i w)+\M[F_{\text{mix}}](2\u-1+\i(v-w)),
%\end{multline*}
%we further have
%\begin{multline*}
%\E\left[\Lambda_1 (X_1, x)\overline{\Lambda_1 (X_1, x)}\right] \\ 
%\hspace{0.7cm} \\= \int\limits_{-T} ^T \int\limits_{-T} ^T \frac{\M[F_{\text{mix}}](2\u-1+\i(v-w))}{(\u-1+\i v)(\u-1-\i w)\M[G](\u+\i v)\overline{\M[G](\u+\i w)}}\\ 
%\hspace{4.5cm} \times x^{2(-\u+1)-\i(v-w)}\left(1-\frac{|v|}{T}\right)\left(1-\frac{|w|}{T}\right)\,dv\,dw \\ 
%- \int\limits_{-T} ^T \int\limits_{-T} ^T \frac{\M[F_{\text{mix}}](\u+\i v)\M[F_{\text{mix}}](\u-\i w)}{(\u-1+\i v)(\u-1-\i w)\M[G](\u+\i v)\overline{\M[G](\u+\i w)}}\\ 
%\hspace{4.5cm} \times x^{2(-\u+1)-\i(v-w)}\left(1-\frac{|v|}{T}\right)\left(1-\frac{|w|}{T}\right)\,dv\,dw,
%\end{multline*}
%where the subtrahend is a real  positive number. 
%Therefore,
%\begin{multline*}
%\E\left[\Lambda_1 (X_1, x)\overline{\Lambda_1 (X_1, x)}\right] \\ 
%\leq x^{2(-\u+1)} \int\limits_{-T} ^T \int\limits_{-T} ^T x^{-\i(v-w)} \frac{\sqrt{\M[F_{\text{mix}}](2\u-1+\i(v-w))}}{(\u-1+\i v)\M[G](\u+\i v)}  \left(1-\frac{|v|}{T}\right) \\ 
%\times \frac{\sqrt{\M[F_{\text{mix}}](2\u-1+\i(v-w))}}{\overline{(\u-1+\i w)\M[G](\u+\i w)}}\left(1-\frac{|w|}{T}\right)\,dv\,dw
%\end{multline*}
Further, by the Cauchy-Schwarz inequality,
\begin{multline*}
\hspace{-0.3cm}\E\left[\Lambda_1 (X_1, x)\overline{\Lambda_1 (X_1, x)}\right] \leq x^{2(-\u+1)} \int\limits_{-T} ^T \int\limits_{-T} ^T \frac{|\M[F_{\text{mix}}](2\u-1+\i(v-w))|}{((\u-1)^2 + v^2)|\M[G](\u+\i v)|^2} \\ 
 \hspace{9cm} \times\left(1-\frac{|v|}{T}\right)^2\,dv\,dw \\ 
% \hspace{3.6cm} \leq x^{2(-\u+1)} \int\limits_{-T} ^T \Bigl[ \frac{(1-|v|/T)^2}{((\u-1)^2 + v^2)|\M[G](\u+\i v)|^2} \Bigr.\\ 
% \Bigl.\hspace{5.5cm} \times\int\limits_{-\infty} ^{\infty} |\M[F_{\text{mix}}](2\u-1+\i(v-w))|\,dw \Bigr]  \,dv \\ 
 \hspace{3.95cm} \leq x^{2(-\u+1)} \int\limits_{-T} ^T \frac{1}{((\u-1)^2 + v^2)|\M[G](\u+\i v)|^2} \,dv \\
\times\int\limits_{-\infty} ^{\infty} |\M[F_{\text{mix}}](2\u-1+\i w)|\,dw.
\end{multline*}
Combining   the last inequality with~\eqref{bbr}, we arrive at the desired result.
\bibliographystyle{apalike}
   \bibliography{Panov_bibliography}

\end{document}